\ifdefined\WCLSingleColumn

\documentclass[journal,onecolumn,draftcls,12pt]{IEEEtran}
\else
\documentclass[journal]{IEEEtran}
\fi
\usepackage{amsmath,amssymb,amsthm}
\usepackage{graphicx}
\usepackage{booktabs}
\usepackage{cite}
\usepackage{hyperref}

\newtheorem{lemma}{Lemma}
\newtheorem{proposition}{Proposition}
\newtheorem{corollary}{Corollary}
\newcommand{\E}{\mathbb{E}}
\newcommand{\R}{\operatorname{Re}}
\newcommand{\F}{\mathcal{F}}

\newcommand{\wclfigurewidth}{\columnwidth}
\ifdefined\WCLSingleColumn
\renewcommand{\wclfigurewidth}{3.45in}
\fi

\begin{document}

\title{Wait to Learn: Timing the Gradient Return in Over-the-Air Split Learning over an Aging Single-Antenna Channel}

\author{Xianzhong Li, Yuzhi Yang, Tianhao Guo}

\maketitle

\begin{abstract}
In over-the-air split learning over an aging single-antenna Rayleigh channel, the mean returned gradient is the antenna-domain gradient scaled by the temporal correlation. Uncorrected training collapses in negative correlation lobes. Near a correlation zero, no nonzero scalar calibration raises a given transmission's ratio of mean power to innovation power, and on the tested grid sixteen times more pilots recover only 7\% of the training gap. Returning at the first correlation extremum with phase correction, 1.3 ms later at 3.5 GHz and 15 m/s, reaches 90.2\% accuracy versus 91.8\% with CSI. In a tested line-of-sight setting, the near-zero gap became negligible.
\end{abstract}


\section{Introduction}
End-to-end training of a wireless transmitter needs the gradient of a receiver-side loss with respect to the transmitter parameters. Over-the-air split learning obtains it without a channel model: the receiver computes the loss derivative at its antenna and sends it back, and reciprocity with phase conjugation applies the adjoint of the channel in the air~\cite{wang2021relay,yang2022mlsp,yang2022jsac}. The idea, also called physical backpropagation, is shared by in-situ training of photonic networks~\cite{hughes2018,pai2023} and by learning in general physical systems~\cite{bosch2026unifying}. Model-free alternatives instead feed back a scalar loss over a reliable link~\cite{aitaoudia2019}. The reciprocity-based wireless schemes assume that the backward channel equals the forward one, and Yang \emph{et al.}~\cite{yang2022jsac} keep reciprocity within a slot and leave non-strict reciprocity open.

A moving terminal breaks this assumption. Channel aging is described by the temporal correlation of the channel, a Bessel-type function for diffuse scattering~\cite{clarke1968,truong2013aging}. Two classical results apply to this setting. In filtered-reference adaptive control, an update through a mismatched plant model still descends on average in the slow-adaptation limit while its phase error stays below $90^\circ$~\cite{morgan1980,boucher1991}. In stochastic optimization, the descent lemma bounds the progress of a biased, noisy gradient step by its mean and second moment~\cite{bottou2018,ajalloeian2020}. Feedback timing under temporal correlation has been analyzed for limited-feedback beamforming~\cite{huang2009}. We apply both classical results to a single receive antenna and derive three consequences for training.

A movable antenna can undo aging by returning to its forward position~\cite{keeppace}, and a receive array can move the correction into space~\cite{companion}. Under the protocol studied here, without current CSI or extra digital gradient feedback, a fixed single antenna has neither option and chooses only \emph{when} to return.

The contributions are as follows. 1) The mean returned gradient is the antenna-domain gradient rotated by the phase of the correlation, and uncorrected training collapses in negative lobes and revives in positive ones. 2) For a given forward transmission, any nonzero scalar correction built from channel statistics leaves the antenna-domain ratio of mean power to innovation power unchanged. At correlation zeros no mean is left to recover, and on the tested near-zero grid more pilots improve the phase estimate while the training gap stays largely open. 3) Returning at the first correlation extremum, with its millisecond-scale wait counted, recovers most of the loss, and the next two extrema give lower correlation, longer waits and higher cross-entropy. Because the extremum is a stationary point, the wait need not be timed precisely. In a tested line-of-sight setting the near-zero gap with exact phase became negligible, while phase correction still mattered where the correlation had a negative real part. Local one-bit phase correction suffices for real correlation and two bits for a broad asymmetric spectrum on the tested grids.

\section{System Model}\label{sec:model}
A transmitter with real parameters $\theta\in\mathbb{R}^p$ and $N_t$ independent ports maps an input $x$ to $s=f_\theta(x)\in\mathbb{C}^{N_t}$. A single-antenna receiver observes $y=h_f s+n$, $h_f\in\mathbb{C}^{1\times N_t}$, evaluates a loss $\ell$, and forms the seed $\delta=2\,\partial\ell/\partial y^*$. The antenna-domain gradient is $u=h_f^H\delta\in\mathbb{C}^{N_t}$, and the parameter gradient is $g=B(u)$, where $B:\mathbb{C}^{N_t}\to\mathbb{R}^p$ is the real-linear map $B(u)_k=\R(u^H\partial s/\partial\theta_k)$. The receiver returns $\delta$ after power normalization, and the transmitter restores the norm from a control field. With phase conjugation and exact norm restoration, the transmitter receives
\begin{equation}\label{eq:return}
r=h_r^H\delta+w,\qquad w\sim\mathcal{CN}(0,\sigma_w^2 I),
\end{equation}
where $h_r$ is the channel during the return. It may multiply $r$ by a correction $c\in\mathbb{C}$ and updates $\theta$ with $\hat g=B(cr)$. If $h_r=h_f$, $c=1$ and $w=0$, then $\hat g=g$.

Between the forward and backward transmissions the receiver moves by $d$ wavelengths. Both channels come from the same zero-mean proper Gaussian angular field~\cite{clarke1968} with unit-variance entries, so $h_r=\rho h_f+e$ with $e\sim\mathcal{CN}(0,(1-|\rho|^2)I)$ independent of $h_f$, and
\begin{equation}\label{eq:rho}
\rho(d)=\int p(\vartheta)\,e^{\,j2\pi d\cos(\vartheta-\vartheta_m)}\,d\vartheta
\end{equation}
for angular power density $p\ge0$ with unit integral and motion direction $\vartheta_m$. For an isotropic spectrum $\rho=J_0(2\pi d)$, with zeros near $d=0.383$, $0.879$, $1.377$ and extrema at $d_1=0.610$ ($\rho=-0.403$), $1.117$ ($0.300$) and $1.619$ ($-0.250$). A terminal at speed $v$ with return delay $\tau$ has $d=v\tau/\lambda$. At 3.5\,GHz the first zero is reached after $\tau=2.2$\,ms at 15\,m/s and $1.0$\,ms at 33\,m/s, so near-zero returns arise at vehicular speeds and millisecond delays. The Rayleigh taps of the 3GPP TDL models use the classical Doppler spectrum~\cite[Sec.~7.7.2]{tr38901}, so each such tap ages as $J_0$. As an extension, Section~IV adds a line-of-sight component of Rician factor $K$ at angle $\beta$ to the motion, common to all ports and with a uniformly random phase per block. The port-averaged cross-covariance $\rho_K=\sum_t\E[h_{r,t}h_{f,t}^*]/N_t$ then equals $(Ke^{\,j2\pi d\cos\beta}+\rho(d))/(K+1)$, so $|\rho_K|\ge(K-1)/(K+1)$ for $K>1$ and the zeros disappear. The random phase makes this channel a non-Gaussian mixture with correlated ports, so the independent-port Gaussian model and the analysis of Section~III refer to $K=0$, and $\rho_K$ is not the only coefficient of the conditional mean.

Let $\F$ collect everything fixed before the return: $\theta$, $x$, the forward noise, $h_f$, the receiver state, and any calibration from earlier pilots. Then $\delta$, $B$ and $c$ are $\F$-measurable, and $e$ and $w$ are independent of each other and of $\F$. The loss is real-differentiable in $y$.

\section{What Calibration Can and Cannot Do}
\begin{lemma}\label{lem:rot}
$\E[r\mid\F]=\rho^*u$. For $u\neq0$, the negative mean of the uncorrected return is an antenna-domain descent direction if and only if $\R\rho>0$. For $\rho\neq0$ and $c=e^{j\arg\rho}$, $\E[\hat g\mid\F]=|\rho|\,g$ for every real-linear $B$.
\end{lemma}
\begin{proof}
$\E[h_r^H\mid\F]=\rho^*h_f^H$ by independence of $e$ and $\F$, and $\R\langle u,\rho^*u\rangle=\R\rho\,\|u\|^2$. With $c=e^{j\arg\rho}$, $c\rho^*=|\rho|$ is real, and real scaling commutes with $B$.
\end{proof}
The antenna-domain condition is the $90^\circ$ condition of filtered-reference adaptation in the slow-adaptation limit~\cite{morgan1980,boucher1991}. Without correction and with complex $\rho$, the parameter-domain direction of a real-linear $B$ is not fixed by $\arg\rho$ alone.

Write $r=\rho^*u+z$ with $z=e^H\delta+w$. Given $\F$, $z\sim\mathcal{CN}(0,\nu I)$ with $\nu=(1-|\rho|^2)|\delta|^2+\sigma_w^2$. Regard $B$ as a real $p\times2N_t$ matrix acting on stacked real and imaginary parts.
\begin{proposition}\label{prop:inv}
For every $\F$-measurable $c\in\mathbb{C}$,
\begin{align}
\E[\hat g\mid\F]&=B(c\rho^*u),\label{eq:mean}\\
\E\big[\|\hat g-\E[\hat g\mid\F]\|^2\,\big|\,\F\big]&=\tfrac12|c|^2\nu\|B\|_F^2.\label{eq:moments}
\end{align}
For $c\neq0$, $\nu>0$ and $u\neq0$, the antenna-domain ratio $\|\E[cr\mid\F]\|^2/\E[\|cz\|^2\mid\F]=|\rho|^2\|u\|^2/(N_t\nu)$ is the same for every such $c$. For $c\rho^*\neq0$, $\nu>0$ and $\|B\|_F>0$, the parameter-domain ratio of the squared mean to~\eqref{eq:moments} equals $2|\rho|^2\|B(e^{j\psi}u)\|^2/(\nu\|B\|_F^2)$ with $\psi=\arg(c\rho^*)$, so it depends on $c$ only through $\psi$. For a plain stochastic-gradient step $\theta^+=\theta-\eta\hat g$ with $\rho\neq0$, $c=e^{j\arg\rho}$ and an $L$-smooth expected loss $R$,
\begin{multline}\label{eq:descent}
\E[R(\theta^+)\mid\theta]-R(\theta)\le-\eta|\rho|\,\|\nabla R\|^2\\
+\tfrac{L\eta^2}{2}\,\E\big[|\rho|^2\|g\|^2+\tfrac12\nu\|B\|_F^2\,\big|\,\theta\big].
\end{multline}
\end{proposition}
\begin{proof}
$cr=c\rho^*u+cz$, and $cz$ is $\mathcal{CN}(0,|c|^2\nu I)$ given $\F$. A proper complex Gaussian vector with per-entry variance $|c|^2\nu$ has real covariance $\tfrac12|c|^2\nu I_{2N_t}$, which gives~\eqref{eq:mean}--\eqref{eq:moments} and both ratios. Let $R$ be the expected loss over $x$, the forward noise and $h_f$ at fixed receiver parameters, drawn independently of $\theta$, with differentiation and expectation interchangeable. Then $\E[g\mid\theta]=\nabla R(\theta)$. Since $|\rho|$ is fixed by $d$, $\E[\hat g\mid\theta]=|\rho|\nabla R(\theta)$, and \eqref{eq:descent} is~\cite[Lemma~4.2]{bottou2018} with these moments.
\end{proof}
With a unit-modulus correction, statistical pilots act only on $\psi$, through the phase error of the estimate. For $\rho\neq0$, dividing by $\rho^*$ restores the mean and amplifies the innovation power by $1/|\rho|^2$, which is the same trade expressed as a larger effective step. Inverse and rotation-only corrections are compared for receive arrays in~\cite{companion,lessons}. The factor $|\rho|$ and the innovation power $(1-|\rho|^2)|\delta|^2$ are set by the displacement. Restarting a forward and return exchange with the same latency leaves its correlation unchanged, so it does not avoid the near-zero return. The bound~\eqref{eq:descent} analyzes a plain gradient step to isolate the mechanism, while the simulations use Adam~\cite{kingma2015adam}. One-step descent and scalar quadratic stability conditions for movable-antenna compensation are derived in~\cite[Sec.~IV-A]{keeppace}.

\begin{corollary}[Zeros]\label{cor:zero}
If $\rho(d)=0$, every update that is homogeneous and real-linear in $(h_r,h_r^*)$ with $\F$-measurable coefficients has zero conditional mean.
\end{corollary}
\begin{proof}
For $\rho=0$, $h_r=e$, which is zero-mean and independent of $\F$. The $\F$-measurable coefficients move outside the conditional expectation.
\end{proof}
Escaping the corollary therefore requires information or operations outside its class, for example nonlinear use of the return, extra feedback, or current CSI of the return channel.

\emph{Delayed return.} Since only the displacement changes $|\rho|$, the receiver can hold the seed when the natural return delay falls near a zero $d_0$ and send it at the next extremum $d_1$ of $|\rho|$, correcting the phase there. Peaks in the temporal correlation of the output SNR were shown to track link performance and proposed for timing adaptive-modulation feedback and MAC slots~\cite{tomasi2011}. Here the returned quantity is a gradient, and the phase of $\rho$ at the extremum fixes its sign. For the isotropic spectrum $\rho(d_1)$ is real and negative, so the correction is a sign flip. The wait costs $(d_1-d_0)\lambda/v$ seconds per forward and return exchange, $0.227\lambda/v$ from the first zero, which is 1.3\,ms at 15\,m/s and $0.6$\,ms at 33\,m/s. Because $d_1$ is a stationary point of $\rho$, the timing need not be exact: for $J_0$, $|\rho|\ge0.32$ within $\pm0.1$ wavelength of $d_1$, and the sign stays negative between the first two zeros, $j_{0,1}/(2\pi)\approx0.383$ and $j_{0,2}/(2\pi)\approx0.879$. At the tested displacements 0.45, 0.55 and 0.70 inside the window $0.45\le d\le0.70$, phase correction reaches $87.9$--$90.0\%$ in Section~IV. The window lasts $1.43$\,ms at 15\,m/s, a duration of 2.9 slots of 0.5\,ms in the 30-kHz NR numerology, and $0.65$\,ms at 33\,m/s. A residual carrier frequency offset of magnitude below 64\,Hz keeps the added phase drift over a 1.3-ms wait below $30^\circ$, an effect not included in our simulations. Choosing the wait requires the terminal speed and the angular spectrum, or the correlation at the candidate delays. Later extrema offer less: for $J_0$ the extremal $|\rho|$ decays, so $|\rho|^2/(1-|\rho|^2)$ falls from $0.194$ at $d_1$ to $0.099$ and $0.066$ at the next two, while the wait grows.

\emph{Few-bit phase.} With true phase $\phi=\arg\rho$, $\rho\neq0$, $u\neq0$, and a correction $e^{jq}$, antenna-domain mean descent requires $\cos(\phi-q)>0$. If the circular error of the phase estimate is at most $\epsilon$, then $\epsilon+|\mathrm{wrap}(\hat\phi-q)|<\pi/2$ suffices. For real $\rho$, $\phi\in\{0,\pi\}$ and a correctly estimated sign bit is exact. For complex $\rho$, a uniform nearest-neighbor $b$-bit phase quantizer adds at most $\pi/2^b$: one bit can reach $90^\circ$, two bits at most $45^\circ$. Few-bit phase feedback has also been studied for over-the-air computation~\cite{dahl2026}.

\section{Simulations}
\emph{Setup.} We train eight equiprobable messages over four channel uses with $N_t=2$ and a single receive antenna. This end-to-end codebook task is a deliberately small instance of over-the-air split learning, in which the gradient return is the only learned path through the channel. Movable-antenna compensation with 3GPP CDL channels is evaluated in~\cite{keeppace}. The transmitter learns a $2\times4$ codebook per message with unit-norm columns, and the receiver is a softmax classifier on the 16 real features of the temporal Gram matrix of the received block. The forward noise variance is $0.10$. The restored norm $b$ is $|\delta|$ quantized to 16 bits on a logarithmic grid clipped to $[2^{-20},2^{20}]$, and the reverse noise $n_r\sim\mathcal{CN}(0,0.0009I)$ is added before restoration. For $|\delta|\ge10^{-30}$ the implemented return is $\kappa h_r^H\delta+b\,n_r$ with $\kappa=b/|\delta|$, and a zero seed returns zero. Its conditional mean is $\kappa\rho^*u$ and its per-entry noise variance is $\kappa^2(1-|\rho|^2)|\delta|^2+0.0009b^2$, which replace $\rho^*u$ and $\nu$ in~\eqref{eq:mean}--\eqref{eq:moments}, while~\eqref{eq:descent} is stated for exact restoration. The four channel uses of a block share one channel pair, and the seed of each use is normalized and restored separately and returned through the same $h_r$, so~\eqref{eq:mean}--\eqref{eq:moments} hold for each use while the innovations of the four uses are correlated through the shared $e$. Encoder and receiver use Adam with learning rates $0.01$ and $0.03$ and batch size 32. Corrections estimate $\rho$ from independent forward and return pilot pairs (64 unless stated), with the forward noise on the forward pilot and noise variance $0.0009$ on the return pilot. The estimate $\hat\rho$ is the trace of the sample cross-covariance divided by the trace of the noise-debiased forward covariance with a $0.02$ ridge and a $0.02$ eigenvalue floor. Phase correction uses $c=e^{j\arg\hat\rho}$, the 2-bit and 1-bit versions round $\arg\hat\rho$ to the nearest multiple of $\pi/2$ and to $0$ or $\pi$ by the sign of $\R\hat\rho$, and the scalar correction is $c=1/\hat\rho^*$ without a magnitude floor. All corrections are applied in the antenna domain. ``Oracle'' arms use the exact $\rho$ and serve as diagnostics. Over-the-air methods use 4032 training blocks, 32\,448 RF vectors with 64 pilot pairs. The CSI baseline trains with digital gradients computed from pilot-estimated CSI on 5440 blocks, 32\,640 RF vectors, and reaches $91.8\%$ test accuracy. It serves as a digital reference rather than a matched-latency baseline, and it exchanges 1\,392\,640 digital control bits against 262\,144 for the over-the-air methods, counted without protocol overhead.

\emph{Statistics.} Training results use 8 independent seeds and 8192 test blocks. The lobe, near-zero, delay, few-bit and array experiments share one seed group, and the extremum comparison of Table~\ref{tab:wait} uses a second group. Primary endpoints were fixed before computation: displacements are averaged within each seed over prespecified grids, followed by a paired bootstrap over seeds (20\,000 resamples), and brackets are 95\% intervals. Other single-cell accuracies are descriptive. ``1-bit'' and ``2-bit'' refer to uniform quantization of the locally applied phase correction, while the pilot statistics are shared as in the continuous case. A separate antenna-domain check over 16 geometries (isotropic and broad asymmetric spectra, two motion directions, four displacements, 65\,536 draws per cell) finds the conditional-mean direction of the uncorrected return within $2.12$ standard errors of $\cos(\arg\rho)$, as Lemma~\ref{lem:rot} requires.

\emph{Lobes.} For the isotropic spectrum (Fig.~\ref{fig:lobes}), uncorrected training collapses in the first negative lobe ($36.2\%$ at $d=0.55$) and revives in the second positive lobe ($89.4\%$ at $d=1.05$, equal to phase correction to one decimal). Averaged over the negative-lobe grid $\{0.40,\allowbreak 0.45,\allowbreak 0.55,\allowbreak 0.70,\allowbreak 0.85\}$ versus the positive-lobe grid $\{0.90,1.05,1.20,1.35\}$, uncorrected cross-entropy (CE) improves by $0.709$ $[0.636,0.792]$. Phase correction stays between $87.9\%$ and $90.0\%$ from $d=0.45$ to $0.70$, across the lobe around $d_1$, and drops only close to the zeros of $J_0$.

\begin{figure}[t]\centering
\includegraphics[width=\wclfigurewidth]{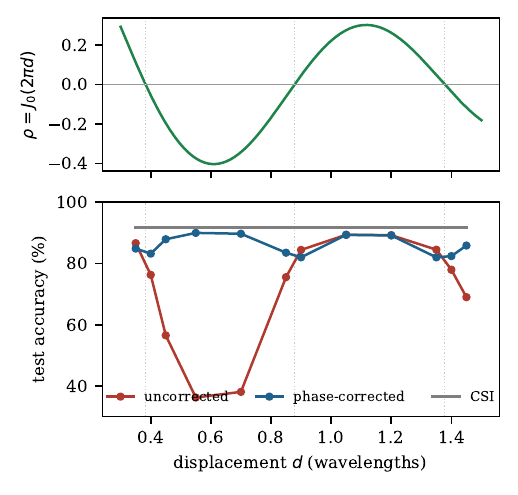}
\caption{Isotropic spectrum, single receive antenna. Top: $\rho=J_0(2\pi d)$, with dotted lines at its zeros. Bottom: test accuracy, mean over 8 seeds.}\label{fig:lobes}
\end{figure}

\emph{Near the first zero.} At $d\in\{0.38,0.40\}$, phase correction with exact $\rho$ is worse than at the reference displacements $\{0.55,0.70\}$ by $0.156$ CE $[0.134,0.178]$. Raising the pilots from 64 to 1024 pairs (35\,328 RF vectors, training blocks fixed) recovers $7.3\%$ of that gap $[-6.4\%,19.7\%]$ (Fig.~\ref{fig:dead}, top), defined as the mean CE improvement from 64 to 1024 pairs near the zero over the mean exact-$\rho$ gap and recomputed in each bootstrap resample. The extra pilots do improve the phase: in 3000 Monte Carlo runs of the same pilot estimator, the mean cosine of the phase error at $d=0.40$ rises from $0.50$ to $0.95$. The factor it multiplies is $|\rho|=0.055$ there and $0.009$ at $d=0.38$, and the exact-$\rho$ oracle, which has no phase error at all, still leaves the gap of $0.156$. Both observations are consistent with Proposition~\ref{prop:inv}. Accuracy near the zero stays around $82$--$84\%$ because the receiver keeps learning: with a frozen random initial transmitter, receiver training alone reaches $83.0\%$, and joint training improves CE over it by only $0.016$ $[0.004,0.029]$ near the zero (exact-$\rho$ correction), against $0.172$ $[0.141,0.206]$ at the reference displacements (64-pair correction). Because the receiver co-adapts with the codebook, this comparison measures the gain of joint training over receiver-only adaptation.

\begin{figure}[t]\centering
\includegraphics[width=\wclfigurewidth]{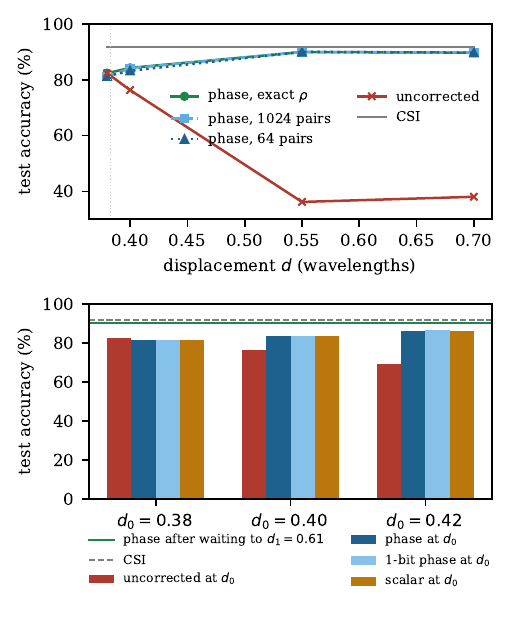}
\caption{Near the first zero of $J_0$, single receive antenna, mean over 8 seeds. Top: phase correction with estimated and exact $\rho$. Bottom: corrections at the natural delay $d_0$ versus phase correction after waiting to $d_1=0.61$.}\label{fig:dead}
\end{figure}

\emph{Delayed return.} Returning at $d_1=0.610$ instead of $d_0\in\{0.38,0.40,0.42\}$ improves CE by $0.153$ $[0.134,0.172]$ over the best correction available at $d_0$, chosen by mean CE among continuous, 2-bit and 1-bit phase and scalar correction (1-bit phase in all three cases), with this selection not included in the interval (Fig.~\ref{fig:dead}, bottom). Waiting reaches $90.2\%$ accuracy against $91.8\%$ for CSI. At $3.5$\,GHz and $15$\,m/s the additional wait is $1.08$--$1.31$\,ms per forward and return exchange. How the waits of the exchanges in a batch overlap depends on the schedule, which we do not model. Table~\ref{tab:wait} repeats the comparison with the second seed group and adds the next two extrema. Among the three tested extrema, the first has the lowest mean CE and the shortest wait: relative to $d_1$, CE rises by $0.017$ $[0.006,0.024]$ at the second and $0.035$ $[0.029,0.040]$ at the third. This ordering is descriptive, and the prespecified margins of $0.02$ and $0.04$ for these two differences were not reached. Against CSI, waiting to $d_1$ remains $0.076$ $[0.064,0.089]$ CE worse.

\begin{table}[t]\centering
\caption{Return delay and training, isotropic spectrum, phase correction with 64 pilot pairs, second seed group of 8. Training grid $d=0.383$, $0.610$, $1.117$, $1.619$ (rounded zero and extremum locations). Additional wait from $d=0.383$ at 3.5\,GHz and 15\,m/s.}\label{tab:wait}
\footnotesize\setlength{\tabcolsep}{4pt}
\begin{tabular}{lccccc}\toprule
Return at $d$ & $\rho$ & $\frac{|\rho|^2}{1-|\rho|^2}$ & wait (ms) & CE & acc.\ (\%)\\\midrule
0.383 (near zero) & $-0.001$ & $<10^{-6}$ & 0 & 0.608 & 82.4\\
0.610 & $-0.403$ & 0.194 & 1.30 & 0.427 & 89.9\\
1.117 & $\phantom{-}0.300$ & 0.099 & 4.19 & 0.443 & 89.2\\
1.619 & $-0.250$ & 0.066 & 7.06 & 0.461 & 88.5\\\midrule
CSI & n/a & n/a & n/a & 0.350 & 91.6\\\bottomrule
\end{tabular}
\end{table}

\emph{Few bits.} Table~\ref{tab:bits} lists accuracies. For the isotropic spectrum, where $\rho$ is real, 1-bit correction is non-inferior to continuous correction over $\{0.55,0.70,1.05,1.20\}$ (CE difference $-0.0012$ $[-0.0026,-0.0001]$, margin $0.02$). For a broad asymmetric spectrum ($70\%$ von Mises at $35^\circ$ with concentration 2, $30\%$ uniform, motion along $\vartheta_m=0$), 2-bit correction is non-inferior when averaged over $\{0.1,\dots,0.5\}$ ($0.0075$ $[0.0061,0.0088]$), while 1-bit correction loses $0.075$ CE $[0.068,0.083]$ at the displacements where $|\arg\rho|$ lies in $[60^\circ,120^\circ]$.

\begin{table}[t]\centering
\caption{Test accuracy (\%) of phase correction with continuous, 2-bit and 1-bit phase, mean over 8 seeds (descriptive)}\label{tab:bits}
\footnotesize
\begin{tabular}{lcccc}\toprule
Spectrum, $d$ & $\arg\rho$ & cont. & 2 bit & 1 bit\\\midrule
isotropic, 0.55 & $180^\circ$ & 90.0 & 90.0 & 90.0\\
isotropic, 1.05 & $0^\circ$ & 89.4 & 89.4 & 89.4\\
broad, 0.30 & $63^\circ$ & 90.8 & 90.6 & 88.8\\
broad, 0.40 & $107^\circ$ & 90.3 & 90.2 & 86.4\\
broad, 0.50 & $148^\circ$ & 90.4 & 89.9 & 89.9\\\bottomrule
\end{tabular}
\end{table}

\emph{Line of sight.} We add the line-of-sight component of Section~II at $\beta=60^\circ$, keeping the task and protocol, with a separate group of 8 seeds. The near-zero gap with exact statistics falls from $0.158$ at $K=0$ to $0.002$ at $K=1$, $0.0004$ at $K=3$ and $-0.00006$ $[-0.00016,0.00005]$ at $K=10$. The means meet the prespecified criteria of no increase in $K$ and a gap below $0.02$ at $K=10$, while the reductions from $K=1$ to 3 and from $K=3$ to 10 are not individually resolved. Phase correction still matters: at $K=10$, averaged over $d\in\{0.55,0.6098,0.70\}$, where the line-of-sight rotation makes $\R\rho_K<0$, uncorrected training is worse than 64-pair phase correction by $0.651$ CE $[0.561,0.735]$.

\emph{Single antenna versus array.} With a 3-element half-wavelength receive array on the same task and number of training blocks, the near-zero gap with exact statistics (an oracle diagnostic) closes when the seed is corrected before the return by the unitary polar factor of the conditional-mean matrix~\cite{companion} ($-0.0001$ $[-0.0005,0.0002]$). A scalar correction by the exact trace correlation, an auxiliary diagnostic, leaves $0.080$ $[0.069,0.092]$. The paired difference between the single-antenna gap and the array gap with polar correction is $0.156$ $[0.134,0.179]$. The difference reflects spatial degrees of freedom, array gain and the matrix correction with its side information together, not antenna count alone. How the element geometry shapes the sign of the empirical correlation across an array is studied in~\cite{lattice}.

\section{Conclusion}
For a single receive antenna, the aging channel scales the mean returned gradient by its temporal correlation. The phase of the correlation fixes the sign of the antenna-domain mean descent, and in our simulations it predicted where uncorrected training fails and revives. For a given forward transmission, no nonzero scalar correction changes the antenna-domain ratio of mean power to innovation power, and on the tested near-zero grid more independent pilots improved the phase estimate but recovered only a small part of the training gap. A fixed single antenna can still use time, returning the gradient within the lobe around the first correlation extremum with a sign correction at a millisecond-scale delay. Waiting matters under diffuse scattering. In the tested line-of-sight setting the near-zero gap with exact phase was negligible, and 64-pair phase correction outperformed uncorrected training where the correlation had a negative real part. The main evaluation assumes Rayleigh fading and the extension a Rician channel with random line-of-sight phase, both with independent realizations, a known correlation function and a synchronous wait. Estimating the speed or correlation online is left to protocol design.

\bibliographystyle{IEEEtran}
\bibliography{refs}

@article{boucher1991,
  author  = {C. C. Boucher and S. J. Elliott and P. A. Nelson},
  title   = {Effect of errors in the plant model on the performance of algorithms for adaptive feedforward control},
  journal = {IEE Proceedings F (Radar and Signal Processing)},
  volume  = {138}, number = {4}, pages = {313--319}, year = {1991},
  doi     = {10.1049/ip-f-2.1991.0042}
}

@article{morgan1980,
  author  = {D. R. Morgan},
  title   = {An analysis of multiple correlation cancellation loops with a filter in the auxiliary path},
  journal = {IEEE Transactions on Acoustics, Speech, and Signal Processing},
  volume  = {28}, number = {4}, pages = {454--467}, year = {1980},
  doi     = {10.1109/TASSP.1980.1163430}
}

@article{yang2022jsac,
  author  = {Yuzhi Yang and Zhaoyang Zhang and Yuqing Tian and Zhaohui Yang and Chongwen Huang and Caijun Zhong and Kai-Kit Wong},
  title   = {Over-the-Air Split Machine Learning in Wireless {MIMO} Networks},
  journal = {IEEE Journal on Selected Areas in Communications},
  volume  = {41}, number = {4}, pages = {1007--1022}, year = {2023},
  doi     = {10.1109/JSAC.2023.3242701}
}

@inproceedings{yang2022mlsp,
  author    = {Yuzhi Yang and Zhaoyang Zhang and Zhaohui Yang},
  title     = {Over-the-Air Split Learning with {MIMO}-Based Neural Network and Constellation-Based Activation},
  booktitle = {Proc. IEEE 32nd Int. Workshop on Machine Learning for Signal Processing (MLSP)},
  address   = {Xi'an, China}, pages = {1--6}, year = {2022},
  doi       = {10.1109/MLSP55214.2022.9943518}
}

@inproceedings{wang2021relay,
  author    = {Rui Wang and Yi Jiang and Wei Zhang},
  title     = {A Distributed {MIMO} Relay Scheme Inspired by Backpropagation Algorithm},
  booktitle = {Proc. IEEE Global Communications Conference (GLOBECOM)},
  pages     = {1--6}, year = {2021}, doi = {10.1109/GLOBECOM46510.2021.9685608}
}

@article{hughes2018,
  author  = {Tyler W. Hughes and Momchil Minkov and Yu Shi and Shanhui Fan},
  title   = {Training of photonic neural networks through in situ backpropagation and gradient measurement},
  journal = {Optica}, volume = {5}, number = {7}, pages = {864--871}, year = {2018},
  doi     = {10.1364/OPTICA.5.000864}
}

@article{pai2023,
  author  = {Sunil Pai and Zhanghao Sun and Tyler W. Hughes and Taewon Park and Ben Bartlett and Ian A. D. Williamson and Momchil Minkov and Maziyar Milanizadeh and Nathnael Abebe and Francesco Morichetti and Andrea Melloni and Shanhui Fan and Olav Solgaard and David A. B. Miller},
  title   = {Experimentally realized in situ backpropagation for deep learning in photonic neural networks},
  journal = {Science}, volume = {380}, number = {6643}, pages = {398--404}, year = {2023},
  doi     = {10.1126/science.ade8450}
}

@article{aitaoudia2019,
  author  = {Fay\c{c}al {Ait Aoudia} and Jakob Hoydis},
  title   = {Model-free Training of End-to-end Communication Systems},
  journal = {IEEE Journal on Selected Areas in Communications},
  volume  = {37}, number = {11}, pages = {2503--2516}, year = {2019},
  doi     = {10.1109/JSAC.2019.2933891}
}

@article{clarke1968,
  author  = {R. H. Clarke},
  title   = {A statistical theory of mobile-radio reception},
  journal = {Bell System Technical Journal}, volume = {47}, number = {6}, pages = {957--1000}, year = {1968},
  doi     = {10.1002/j.1538-7305.1968.tb00069.x}
}

@article{truong2013aging,
  author  = {Truong, Kien T. and Heath, Jr., Robert W.},
  title   = {Effects of Channel Aging in Massive {MIMO} Systems},
  journal = {Journal of Communications and Networks}, volume = {15}, number = {4}, pages = {338--351},
  month   = aug, year = {2013}, doi = {10.1109/JCN.2013.000065}
}

@inproceedings{kingma2015adam,
  author    = {Diederik P. Kingma and Jimmy Lei Ba},
  title     = {Adam: A Method for Stochastic Optimization},
  booktitle = {Proc. Int. Conf. Learning Representations (ICLR)}, year = {2015}
}

@misc{bosch2026unifying,
  author = {Cyrill B\"osch and Yi\u{g}ithan Gediz and Clara C. Wanjura and Hakan E. T\"ureci},
  title  = {Unifying Physical Backpropagation},
  howpublished = {arXiv:2608.11585v2}, year = {2026}
}

@misc{companion,
  title = {Coherent but Discordant: Rotate, Don't Invert, to Backpropagate Through an Aging Channel},
  howpublished = {arXiv:2610.23137v1},
  year  = {2026}
}

@misc{keeppace,
  title = {Keep Pace to Learn: Movable Antennas Restore the Backpropagated Gradient Under User Mobility},
  howpublished = {arXiv:2610.23118v1},
  year  = {2026}
}

@misc{lattice,
  title = {Deep Holes and Aperture: Sign Laws for the Empirical Correlation of Lattice-Sampled Random Wave Fields},
  note  = {Companion manuscript},
  year  = {2026}
}

@misc{lessons,
  title = {Negative Results for Reciprocity-Based Over-the-Air Split Learning Under Channel Aging: Stability, Alignment, and Geometry},
  note  = {Preprint},
  year  = {2026}
}

@article{bottou2018,
  author  = {L{\'e}on Bottou and Frank E. Curtis and Jorge Nocedal},
  title   = {Optimization Methods for Large-Scale Machine Learning},
  journal = {SIAM Review},
  volume  = {60}, number = {2}, pages = {223--311}, year = {2018},
  doi     = {10.1137/16M1080173}
}

@misc{ajalloeian2020,
  author = {Ahmad Ajalloeian and Sebastian U. Stich},
  title  = {On the Convergence of {SGD} with Biased Gradients},
  howpublished = {arXiv:2008.00051v2}, year = {2021}
}

@inproceedings{tomasi2011,
  author    = {Beatrice Tomasi and James Preisig and Michele Zorzi},
  title     = {On the Predictability of Underwater Acoustic Communications Performance: The {KAM11} Data Set as a Case Study},
  booktitle = {Proc. 6th ACM Int. Workshop on Underwater Networks (WUWNet)},
  pages     = {1--7}, year = {2011},
  doi       = {10.1145/2076569.2076570}
}

@misc{dahl2026,
  author = {Martin Dahl and Erik G. Larsson},
  title  = {Hybrid Channel Estimation with Quantized Phase Feedback for Over-the-Air Computation},
  howpublished = {arXiv:2601.16054}, year = {2026}
}

@article{huang2009,
  author  = {Huang, Kaibin and Heath, Jr., Robert W. and Andrews, Jeffrey G.},
  title   = {Limited Feedback Beamforming Over Temporally-Correlated Channels},
  journal = {IEEE Transactions on Signal Processing},
  volume  = {57}, number = {5}, pages = {1959--1975}, year = {2009}
}

@techreport{tr38901,
  author      = {{3GPP}},
  title       = {Study on Channel Model for Frequencies from 0.5 to 100 {GHz}},
  institution = {3rd Generation Partnership Project},
  type        = {Technical Report},
  number      = {TR 38.901 V19.5.0},
  year        = {2026}
}

\end{document}